%% file: main.tex
\documentclass[11pt]{article}

\input{macros}

\title{Sensitivity and Computational Complexity in Financial Networks}
\author[1]{Brett Hemenway%
	\thanks{\texttt{fbrett@cis.upenn.edu}; 215-746-1738; Corresponding author\\ Supported in part by a grant from the Jacobs-Levy Foundation. }}
\author[1]{Sanjeev Khanna%
	\thanks{ \texttt{sanjeev@cis.upenn.edu}\\ Supported in part by National Science Foundation grants CCF-1552909 and CCF-1617851, and a grant from the Jacobs-Levy foundation..}}

\affil{University of Pennsylvania\\Department of Computer Science\\Levine Hall 3330 Walnut Street, Philadelphia, PA 19104}

\begin{document}
	
	\maketitle

	\vspace{-2em}
	\input{abstract}

	\bigskip
	\bigskip
	\textbf{Keywords:} Financial Contagion; Computational Complexity; Network Analysis; Network Stability; Sensitivity

	\newpage

	\input{introduction}

	\input{prior_work}

	\input{preliminaries}

	\input{sensitivity}

	\input{complexity}

	\input{conclusion}

	\bibliographystyle{alpha}
	\bibliography{banking}

	\newpage

	\begin{center}
		{\Large \bf
			Appendix
		}
	\end{center}

	\input{appendix}

\end{document}

%% file: macros.tex
\usepackage[margin=1in]{geometry}
\usepackage{amsmath,amssymb,amsthm}
\usepackage{xcolor}
\usepackage{nicefrac}
\usepackage{ifthen}
\usepackage[affil-it]{authblk}
\usepackage{setspace}
\newcommand{\mm}[1]{\ensuremath{\mathbf{#1}}}
\newcommand{\vv}[1]{\ensuremath{\vec{#1}}}
\newcommand{\R}{\ensuremath{\mathbb{R}}}
\newcommand{\half}{\ensuremath{\frac{1}{2}}}
\newcommand{\defined}{\ensuremath{\stackrel{\textrm{def}}{=}}}
\newcommand{\ie}{\emph{i.e.},~}
\newcommand{\eps}{\ensuremath{\epsilon}}
\newcommand{\inparen}[1]{\ensuremath{\left( #1 \right)}}

\usepackage{accents}
\newcommand\munderbar[1]{\underaccent{\bar}{#1}}

\usepackage{tikz}
\usetikzlibrary{calc}

\theoremstyle{definition}
\newtheorem{theorem}{Theorem}
\newtheorem{lemma}{Lemma}
\newtheorem{corollary}{Corollary}

\newtheorem{definition}{Definition}

\usepackage{hyperref}

\hypersetup{
	colorlinks=true,
	linkcolor=black,
	citecolor=black,
	urlcolor=black
}

%% file: abstract.tex
	\begin{abstract}
		Determining the causes of instability and contagion in financial networks is necessary to inform policy and 
		avoid future financial collapse.  

		In the American Economic Review, Elliott, Golub and Jackson proposed a simple model for capturing the 
		dynamics of complex financial networks.  In Elliott, Golub and Jackson's model, 
		the institutions in the network are connected by linear dependencies (cross-holdings) and if any institution's value drops below 
		a critical threshold, its value suffers an additional failure cost.
		This work shows that even in this simple model there are fundamental barriers to 
		understanding the risks that are inherent in a network.

		First, if institutions are not required to maintain a minimum amount of self-holdings, any change in investments 
		by a single institution can have an arbitrarily magnified influence on the net worth of the institutions in the system.  
		This implies that if institutions have small self-holdings, then estimating the market value of an institution requires almost perfect information 
		about \emph{every cross-holding in the system}.

		Second, Even if a regulator has complete information about all cross-holdings in the system, it may be computationally intractable 
		to estimate the number of failures that could be caused by a small shock to the system.
	\end{abstract}

%% file: introduction.tex
	\section{Introduction}

		The recent financial crisis and subsequent bailout have highlighted the need for a better understanding of the dynamics of 
		financial networks.  Indeed, the complexity of modern financial networks has been blamed for our collective failure to recognize the presence 
		of serious risks in these systems.
		In this work, we show that even in extremely simple financial networks understanding the risk present in the system is computationally intractable in general -- 
		even in the presence of perfect information about all participants in the system.  We hope that these insights will help regulators 
		and policymakers to better understand the dynamics of financial networks.

		Understanding an individual's risk in a financial network is a difficult task, because an institution's ability to fulfill 
		its outgoing financial obligations is not a local property, \ie it cannot be understood by examining a single individual in isolation.
		An institution's ability to make its outgoing payments may depend on whether its incoming payments are made by its debtors, which in turn may depend on whether the incoming obligations are made to those institutions, 
		etc.  Thus each institution, acting without a global view of the network, cannot effectively understand its risk.
		Nevertheless, one might hope that a regulator, with a global view of the network could better understand the opportunities and risks inherent in the system.
		
		Even with a global view, the situation remains relatively complex.
		One complicating factor is the existence of cycles in the financial network, for example institution $A$ may have obligations to institution $B$
		which has obligations to institution $C$, which in turn has obligations to institution $A$.  The existence of these cyclical interdependencies has 
		been put forward as one of the primary sources of complexity into financial networks.
		In this work, we will examine network dynamics in both cyclic (Theorem \ref{thm:cyclicsensitivity}) and acyclic networks (Theorem \ref{thm:hardness}).

		One of the driving forces in the study of financial networks is their ability to magnify risk:
		if an institution, $A$, defaults on its obligations to institution $B$, this may cause institution $B$ to default on its obligations to institution $C$ etc.  
		The spread of risk through a financial network is known as \emph{financial contagion} and has been carefully modeled and studied \cite{AG00,EN01,GHM12,AOT13,GY14,EGJ14}.

		This work focuses on quantifying the ability of financial networks to amplify and conceal risk.

	\section{Our contributions}
		In this work, we study two questions related to the stability of financial networks.
		First, we look at how sensitive market valuations can be to small changes in network structure.
		Second, we examine the computational complexity of determining how far a given network is from a massive failure.
		Throughout this work we use the network model put forward by \cite{EGJ14}.
		In this model, financial institutions own shares of underlying assets, and the institutions are connected by cross-holdings, 
		which are modeled as linear dependencies.
		If an institution's market value drops below a certain critical threshold, its value suffers a further discontinuous shock, 
		modeling the effects of a loss of investor confidence, or the failure to pay everyday operating costs.
		See \cite{EGJ14} for an in-depth discussion of the interpretation and real-world validity of this model.
		The formal mathematical model is described in detail in Section \ref{sec:model}.

		Our first result (Theorem \ref{thm:cycliclowerbound}), shows that financial networks can be highly sensitive to small changes 
		in their link structure.  Concretely, we show that if a single institution changes a single cross-holding by $\eps$, the market 
		values of institutions in the system can change by as much as $\eps/2r$, where $r \le 1$ is the minimum amount of self-holdings of the institutions in the network.
		The minimum self-holdings, $r$, is a measure of \emph{integration} of the network, where $r=0$ corresponds to a fully integrated network, 
		and $r = 1$ corresponds to a network with no integration.
		This result shows that if each institution retains only, say, $5\%$ self-ownership, a change in a single holding by $\eps$ can result in a $20 \eps$ change in 
		an institution's market value.  This amplification is directly caused by cycles in the network, and in acyclic networks this type of magnification cannot occur (see Lemma \ref{cor:acyclicsensitivity}).
		Our bounds are essentially tight, and Theorems \ref{thm:cyclicsensitivity} and \ref{thm:cycliclowerbound}, show that the true sensitivity is in fact $\Theta\inparen{ \eps/r}$.
	
		This sensitivity magnification has many consequences.  First, it means that in order to estimate market values of institutions within the system, all cross-holdings need to be known to 
		an extremely high degree of precision.  Investors or regulators who wish to calculate market values can be extremely far off if even a single cross-holding in the network remains unknown to them.
		Second, because small changes in a single institution's investments can have large effects on market values throughout the system, this indicates a potential for extreme instability in the 
		system as a whole; the small portfolio changes in one institution can have a drastically magnified effect market values of other institutions, and so small changes by a myopic institution 
		could topple even seemingly stable institutions.  Third, this extreme sensitivity means that calculating market values in a privacy-preserving manner can be extremely difficult \cite{NPH14}.
		This is the flip-side of the first point, in order to calculate market values for institutions in the system, \emph{all} the interbank holdings need to be known with high precision, 
		which means that revealing market prices has the potential to reveal extremely detailed information about each institution's interbank holdings.


		Our second result addresses the question of how well a regulator can assess the stability of a network.
		Suppose a regulator or oversight agency is presented with a network in which every bank is solvent, and 
		suppose the regulator believes that the underlying assets cannot drop in value by more than some fixed amount $d$.
		What is the maximum number of failures that can be caused by a drop of this size in asset values?
		We emphasize that in this scenario, the regulator has complete information about the entire structure of the financial network, 
		and the only uncertainty is in which specific assets may decline in value.  
		We show calculating the number of institutions that can fail in a network
		is NP-Hard.  In particular, we show that is as hard as calculating the maximum balanced clique in a bipartite graph 
		(Theorem \ref{thm:hardness}).  The maximum balanced bipartite subgraph (BCBS) problem is NP-hard \cite{GJ79,J87}, 
		and there is evidence that it is even hard to approximate.
		It is known that if $3$-SAT is not in $DTIME\inparen{ 2^{n^{3/4 + \eps}} }$ for some $\eps$, 
		then there is no polynomial time algorithm for calculating the maximum balanced clique to within a factor of $2^{(\log n)^\delta}$
		for some $\delta > 0$.
		\cite{FK04} go further, conjecturing that there is no polynomial time 
		algorithm to approximate the maximum balanced bipartite clique to within a factor of $n^\delta$ 
		for some $\delta > 0$.
		This would imply that there is no polynomial time algorithm that can even \emph{estimate} the maximum number of failures caused by a drop in asset values of 
		a given magnitude to within a factor of $n^\delta$ in a network with $2n$ institutions.  In particular, this means that there are financial networks where stress testing 
		(which is inherently computationally feasible) cannot hope to even approximate the magnitude of collapse that could be caused by some bounded drop in asset values.

		Unlike our first result, which crucially relies on cycles in the network, this result holds even in acyclic networks: 
		even when there are no cycles it is computationally intractable to estimate the maximum number of failures that could be caused 
		by a bounded drop in asset values.  This complexity arises not from cycles, but from the nonlinear dynamics that occur when 
		a bank drops below its critical threshold value.  

%% file: prior_work.tex
\section{Previous Work}

	Many different models have been proposed to study financial networks, and specifically models of stability and contagion.

	Allen and Gale \cite{AG98} considered a model consisting of depositors and banks.  Depositors deposit their money in the banks, and the banks 
	must choose between making \emph{short-term} investments or \emph{long-term} investments.  This model has three time periods, $t=0,1,2$.
	At $t=0$, investments are made, at $t=1$ short-term investments pay off, and depositors choose whether to withdraw their money, and at time $t=2$ 
	long-term investments pay off.  The banks' investment strategies then depend on the probability that depositors withdraw their money at time $t=1$.
	In a follow-up work \cite{AG00} introduced a network component, whereby banks can exchange deposits 
	with each other in an effort to mitigate risk, and they showed simple contagion effects in this model.

	Acemoglu, Ozdaglar and Tahbaz-Salehi \cite{AOT13} build on Allen and Gale's 3 time-step model.
	At time $t=0$ banks can make a short-term investment, a long-term investment or loan money to other banks.
	Long term investments yield a fixed return at $t=2$.
	Long term investments that are liquidated at $t=1$ receive a return that is randomly 
	distributed between two values.
	Banks whose investments returned the lower of the two values were said to have received a \emph{shock}.  
	Acemoglu, Ozdaglar and Tahbaz-Salehi considered how two extremal types of networks serve to propagate these shocks.
	They considered the ring (where each bank only has debts to its two neighbors) and the complete network where each bank's debts are spread to all other banks, 
	as well as all convex combinations of these two.  They showed that if the magnitude of the shocks are small then the complete network is more stable and resilient than the ring, 
	and if the magnitudes of the shock are large enough then both the ring and complete networks are the \emph{least} stable networks, 
	thus the complete network exhibits a phase transition, moving from the most stable to the least stable network as the magnitude of the shocks increase.

	Eisenberg and Noe \cite{EN01}  developed a very simple and appealing network model where each bank 
	has cash reserves and fixed debts to other banks.  Eisenberg and Noe's work focused on showing that (under some basic restrictions on the network)
	there is always a unique clearing vector (indicating how much of its debts each bank pays to its creditors), and they gave a linear program and simple iterative algorithm 
	for calculating this clearing vector and hence the equilibrium valuation of each bank.

	Gai and Kapadia \cite{GK10} considered a modification of Eisenberg and Noe's network model, forcing all incoming edges to have the same weight, 
	but allowed banks to have additional illiquid assets.
	Gai and Kapadia then considered the question of how a single bank failure propagates through a network.
	For this analysis, they considered different models for generating the underlying graph topology, and plotted contagion effects for different network models 
	characterized by their degree distribution.  They found that a single large shock could have devastating 
	effects on the network, but that this was highly dependent on where in the network the shock hit.

	Gourieroux, Heam and Monfort \cite{GHM12} considered a model that allows interbank investment via shares (like \cite{EGJ14}) and lending or insurance (like \cite{EN01}).
	Unlike \cite{EGJ14}, they do not introduce discontinuous failure costs.  Gourieroux, H{\'{e}}am and Monfort extend Eisenberg and Noe's uniqueness results to show that (under mild constraints on the network) 
	this extended model has a unique equilibrium value for all institutions.  They then examined the effects of exogenous shocks on the network (\ie drops in asset values) 
	using synthetic data and data obtained from the French banking sector.

	Morris \cite{M00} studied contagion in local interaction systems.  In that model, each node represents a player, and each player engages in a local game with 
	each of its neighbors.  Morris focused on the case where each player has a binary strategy space, $\{0,1\}$, and there is a global $2 \times 2$ payoff matrix, such 
	that for each edge in the graph, the two neigboring players receive a payoff according to this global payoff matrix.  In that model, a strategy is said to be 
	``contagious,'' if it can spread from a finite set of players to an infinite set of players by the best-response dynamics of the underlying local game.

	The notion of failure cascades and contagion have also been studied in the computer science literature by \cite{BEKKT11}.
	Blume et al. considered general cascades in graphs where the edges were unweighted and a node was said to fail if some critical threshold of its neighbors failed.
	By choosing all edges to have equal weight, and choosing each institution's failure threshold carefully, the failure model of \cite{EGJ14} can be made to overlap 
	with this general network failure model.

	In this work, we use the model of \cite{EGJ14}.
	In this model, institutions can own shares in each other, or in ``primitive assets'' that have intrinsic value
	outside of the network.  The model is explained in detail in the next section.
	Elliott, Golub and Jackson introduced this model to help analyze and understand contagion effects in networks.
	Because a cascade of collapse requires an initial failure, Elliott, Golub and Jackson began 
	by showing that the weakest institution can never be made strictly more stable by any fair trade between the institutions.
	Next they examined the contagion dynamics in a wide class of networks parametrized 
	by ``integration'' and ``diversification.''  Integration increases as institutions in the network 
	increase their inter-network holdings, \ie integration increases as the percentage of each institution 
	owned by shareholders external to the network decreases.  As integration increases, the institutions 
	fates are more closely tied together.  Diversification measures how risk is spread within the network.
	Diversification increases as institutions increase their \emph{number} of cross-holdings.
	Neither integration nor diversification have strictly positive or negative effects on network stability, but 
	instead have slightly more complex non-monotonic effects.

	The notion of computational complexity has been studied in the context of financial products by Arora et al. who 
	showed that banks can create derivatives that are computational intractable to price accurately \cite{ABBG10,ABBG11}.
	The result of Arora et al. crucially relies on the information asymmetry between the institution the seller (who creates 
	the derivatives) and the buyer who only sees their resulting composition.  Braverman and Pasricha \cite{BP14} show 
	that even in the full information setting pricing compound options is PSPACE complete.

%% file: preliminaries.tex
\section{Model}
\label{sec:model}

We use the model put forward by \cite{EGJ14}.
In this model there are $n$ financial institutions, these can be viewed as countries, banks or private firms, 
and $m$ underlying assets, that can be viewed as any object or project with intrinsic value.
The financial institutions own shares of the underlying assets which impart value into the system.
The values of the institutions themselves are interconnected via a network (modeled as a weighted, directed graph).
The interdependencies (cross-holdings) between the institutions are modeled as simple linear dependencies.
These linear cross-holdings can model simple equity stakes (one institution owning shares in another) or they can be 
viewed as an approximation of more complicated debt contracts between the institutions 
(see \cite[Section 2.5]{EGJ14} for a more thorough discussion of the validity and generality of this model).

Although institutions invest in one another, all value in the system originates from the underlying assets.
The price of asset $k$ is denoted by $p_k$, and we use $D_{ik} \ge 0$ to denote the percentage of asset $k$ 
owned by institution $i$.  The $n \times m$ matrix of ownership is denoted by $\mm{D} = (D_{ik})$.

We define $\mm{C} = (C_{ij})$ to be the $n \times n$ matrix indicating the cross-holdings of institutions.
Thus institution $i$ owns a $C_{ij}$ fraction of institution $j$.  It will be be useful to view the network 
of cross-holdings as a directed graph with $n$ nodes representing the financial institutions, and 
an edge from institution $j$ to $i$ of weight $C_{ij}$ whenever $C_{ij} > 0$.  Following \cite{EGJ14}, 
we set $C_{ii} = 0$ for all $i$.  Now, $\sum_{i} C_{ij}$ is the fraction of institution $j$ that is owned 
by institutions external to $j$.  The remainder, the amount of self-ownership, is denoted by $\hat{C}_{jj} \defined 1 - \sum_{i} C_{ij}$.
The matrix $\hat{\mm{C}}$ will be a diagonal matrix with $\hat{C}_{ii}$ on the diagonal.

As noted by \cite{BBC89}, this type of model introduces two types of valuations, 
the \emph{equity valuation} ($\vv{V}$) and the \emph{market valuation} ($\vv{v}$).  
The equity valuation of institution $i$ is denoted by 
\begin{equation}
	\label{eqn:equityval}
		V_i = \underbrace{ \sum_{k} D_{ik} p_k }_{\mbox{Value of assets held by $i$}} + \underbrace{\sum_j C_{ij} V_j}_{\mbox{Values of institutions held by $i$}}
\end{equation} 

In matrix notation, this becomes $\vv{V} = \mm{D} \vv{p} + \mm{C} \vv{V}$
which implies $\vv{V} = (\mm{I} - \mm{C})^{-1} \mm{D} \vv{p}$.
The matrix $\mm{I} - \mm{C}$ is guaranteed to be invertible because we assume that $\hat{C}_{jj} > 0$, so the column sums of $\mm{C}$ are all strictly less than one (see Lemma \ref{lem:matrixsum}).  
In fact, the matrix $\mm{I} - \mm{C}$, is an M-Matrix \cite{PB74}, and so $(\mm{I} - \mm{C})^{-1}$ is an inverse M-Matrix, 
about which many properties are known \cite{W77,J82}.

This equity valuation significantly overvalues the institutions.  In particular, we can see that $\|\vv{V}\|_1 \ge \|\vv{p}\|_1$, so 
the total value of the institutions in the system will (in general) be much larger than the total value of the underlying assets.  This occurs because each 
asset counts towards the equity value of the institution that owns it and also to the institutions that have an equity stake in the asset's owner.
The network's inflation of equity values is well-known and validated both theoretically and empirically \cite{FP91,FHT94}.

To find an institution's \emph{market value}, we must scale the institution's equity value by the percent stake it has in itself, thus the market value of institution $i$ is
$v_i = \hat{C}_{ii} V_i$, so the market values are the solution to the system
\begin{equation}
	\label{eqn:marketvaluation}
	\vv{v} = \hat{\mm{C}} \vv{V} = \hat{\mm{C}} (\mm{I} - \mm{C})^{-1} \mm{D} \vv{p}
\end{equation}

The matrix $\mm{C}$ is column sub-stochastic because column $i$ sums to $1 - \hat{C}_{ii}$.
The system can also be viewed as a flow, where at each time step money flows between banks according to the link structure of the network (see Appendix \ref{app:flow}).

%% file: sensitivity.tex
\section{Sensitivity}

	Our first result concerns concerns the sensitivity of valuations to small changes in the structure of the network.
	Suppose a single institution shifts its holdings by a small quantity, $\eps$, how much can this small change affect the market valuations 
	in the network?  This question is motivated by questions about network stability, and the possibility of privacy-preserving oversight.
	If small changes in network holdings can lead to large changes in the market values of the institutions, this indicates a fundamental instability in the 
	financial network.  
	Additionally, if small changes in interbank holdings can lead to large changes in market values, then any attempt at financial oversight 
	must know all the interbank holdings to a high degree of accuracy in order to predict market values.  

	A high sensitivity also has implications towards privately computing network statistics.  
	Flood et al. \cite{FKOS13} proposed using tools from 
	differential privacy \cite{D06,DMNS06} to provide a means of computing global network characteristics while preserving the privacy of 
	each individual institution's holdings.  A high sensitivity implies a worse trade-off between privacy and accuracy when calculating network statistics.

	We begin our sensitivity analysis with a simple observation: if the total value of the underlying assets is $\|\vv{p}\|$, then an $\eps$ change in holdings can easily change the 
	market valuations by $\eps \|\vv{p}\|$ (see Figure \ref{fig:simplechange}).  
	We show that if the network is acyclic, then this is the largest change possible, but if there are cycles in the network, the sensitivity can be much larger.

	\begin{figure}
		\begin{center}
			\tikzstyle{asset}=[circle,thick,draw=red!50!black]
			\tikzstyle{bank}=[circle,thick,draw=blue]
			\tikzstyle{shadow}=[circle,thick,draw=green!50!black]
			\begin{tikzpicture}
				\node (A) [asset] at (0,2) {$1$};
				\node (B1) [bank] at (0,0) {$B_1$};
				\node (B2) [bank] at (2,0) {$B_2$};
				\node (S1) [shadow] at ([yshift=-2cm]B1) {};
				\node (S2) [shadow] at ([yshift=-2cm]B2) {};

				\draw [->] (A) to node [right] {1} (B1);
				\draw [->,bend right] (B1) to node [right] {1} (S1);
				\draw [->] (B2) to node [right] {1} (S2);
		
				\draw [->,bend left,dashed] (B1) to node [right] {$1-\eps$} (S1);
				\draw [->,dashed] (B1) to node [above] {$\eps$} (B2);
			\end{tikzpicture}
		\end{center}
		\caption[An $\eps$ change results in an $\eps \|\vv{p}\|$ change in market values.]{%
			An $\eps$ change results in an $\eps \|\vv{p}\|$ change in market values. 
			Banks are in blue, external shareholders are in green, and the asset is shown in red.
			In this example $\mm{D}\vv{p} = \left[ \begin{array}{c} 1 \\ 0 \end{array} \right]$, and $\mm{C} = \mm{0}$.
			Thus $\vv{v} = \left[ \begin{array}{c} 1 \\ 0 \end{array} \right]$.
			Changing $\mm{C}$ to $\tilde{\mm{C}} = \left[ \begin{array}{cc}	0 & \eps \\ 0 & 0 \end{array} \right]$ 
			leads to a valuation of $\tilde{\vv{v}} = \left[ \begin{array}{c} 1-\eps \\ \eps \end{array} \right]$, 
			thus the market valuation of $B_2$ changes by $\eps \|\vv{p}\|$.
		\label{fig:simplechange}}
	\end{figure}

	Throughout this section, we use $r$ (for ``reserve'') to measure the fraction of each institution held by investors \emph{outside} the system.\footnote{%
	The reserve, or self-holdings, can be viewed as the amount of an institution that is not sold, or is held by private shareholders, who retain complete 
	ownership of themselves.  These private shareholders buy shares of institutions in the network, but no entity in the network owns shares of the private shareholders.}~%
	We define $r = \min_i \hat{C}_{ii}$.
	Using the terminology of \cite{EGJ14}, $r$ is just a concrete metric of the \emph{integration} of the network, and integration increases as $r \rightarrow 0$.		
	Another interpretation of reserve is the discrepancy between equity valuation ($\vv{V}$) and market valuation ($\vv{v}$).  Since $\vv{v} = \hat{C} \vv{V}$, we 
	have that $r = \min_i \frac{v_i}{V_i}$.

	\subsection{Sensitivity in acyclic networks}

	We begin by noting that in the acyclic case, there is a strong bound on each institution's equity valuation, \ie the equity 
	valuation cannot be too much larger than the market valuation.
	\begin{lemma}
		\label{lem:acyclicequity}
		If the banking network has no cycles, then every institution's equity valuation is at most $\|\vv{p}\|_1$ where $\vv{p}$ is the vector of asset values.
	\end{lemma}

	\begin{proof}
		Organize the financial network into layers, so that each institution only owns shares of institutions at lower layers.
		Thus institutions at level 1 do not have any cross-holdings, they only own the underlying assets.
		This means that the incoming edges to level one carry a total weight of at most $\|\vv{p}\|_1$.
		By adding fictitious institutions that pass on all of their incoming wealth, we can also ensure that every institution 
		only owns shares in the entities at the preceding level, \ie institutions at level $i$ only own shares in institutions (or fictitious institutions) 
		at level $i-1$.
	
		Now each institution's equity value is the sum of the values on all incoming edges.  Since outgoing edges carry a value that is a percentage of equity value, 
		the sum of the values on each institution's outgoing edges is at most the sum of the values on its incoming edges.  (For real institutions the outgoing sum will be strictly less 
		because the reserve rate $r > 0$, but for the fictitious institutions it can be exactly equal.)

		Now, the sum of the values coming into layer 1 is at most the sum of the assets, $\|\vv{p}\|_1$.  Thus the outgoing edges from level 1 to level 2 carry a total weight of at most $\|\vv{p}\|_1$.
		Proceeding inductively through the levels, we see that the sum of the values on the incoming edges at level $i$ is at most $\|\vv{p}\|_1$.  Thus the equity value of all the institutions on level $i$ is at most 
		$\|\vv{p}\|_1$ and in particular the equity value of any given institution is at most $\|\vv{p}\|_1$.
	\end{proof}
	
	\begin{corollary}
		\label{cor:acyclicsensitivity}
		If the banking network is acyclic, and one edge changes by at most $\epsilon$, then no institution's market value can change by more than $\epsilon \|\vv{p}\|_1$.
	\end{corollary}

	\begin{proof}
		Since each institution's equity value is at most $\|\vv{p}\|_1$, an $\epsilon$ change in any edge corresponds to an absolute change of at most $\epsilon \|\vv{p}\|_1$.
	\end{proof}

	\subsection{Sensitivity in general networks}

	In this section, we explore how much the market valuations can change when one bank changes its holdings by a small amount 
	\emph{in the presence of cycles in the network graph}.	
	We begin by showing an upper bound on the change in market valuations that depends on the minimum self-ownership ($\hat{C}_{ii}$) of the institutions.
	For our upper bound, we do not require changes to occur in the holdings of a single bank.  Instead, we allow any change in network structure, 
	as long as the total ($\ell_1$) change is bounded by $\eps$.
	Formally, this means that we have two network matrices $\mm{C}$ and $\tilde{\mm{C}}$
	such that $\|\mm{C} - \tilde{\mm{C}}\| \le \eps$, and we would like to bound how much the 
	market valuations can change between these two situations.
	Because we are showing an upper bound, allowing more general perturbations only strengthens our result.

	\begin{theorem}
		\label{thm:cyclicsensitivity}
		If $\|\mm{C} - \tilde{\mm{C}}\| < \eps$, then $\|\vv{v} - \tilde{\vv{v}} \| < \frac{\eps}{r} \left\| \mm{D} \vv{p} \right\|$, where 
		$r = \min_{i}( \tilde{\hat{C}}_{ii},\hat{C}_{ii} )$ is the minimum reserve or ``self-holdings'' of the financial institutions.
		In addition, it is always true that $\|\vv{v} - \tilde{\vv{v}} \| \le 2 \left\| \mm{D} \vv{p} \right\|$, thus 
		\[
			\frac{ \|\vv{v} - \tilde{\vv{v}} \|}{\|\mm{D}\vv{p}\|} \le \min \inparen{ \frac{\eps}{r}, 2 }
		\]
	\end{theorem}

	\begin{proof}	
		Let $\tilde{\mm{C}} = \mm{C} + \mm{E}$, and $\hat{\tilde{\mm{C}}} = \hat{\mm{C}} + \hat{\mm{E}}$.
		By hypothesis $\| \mm{E} \| + \| \hat{\mm{E}} \| < \eps$.
		Then we have

		\begin{align*}
			\hat{\tilde{\mm{C}}}(\mm{I} - \tilde{\mm{C}})^{-1} - \hat{\mm{C}}(\mm{I} - \mm{C})^{-1} 
				&= \left[ (\hat{\mm{C}} + \hat{\mm{E}}) - \hat{\mm{C}}(\mm{I}-\mm{C})^{-1}(\mm{I} - \tilde{\mm{C}}) \right] (\mm{I} - \tilde{\mm{C}})^{-1} \\
				&= \left[ (\hat{\mm{C}} + \hat{\mm{E}}) - \hat{\mm{C}}(\mm{I}-\mm{C})^{-1}(\mm{I} - \mm{C}-\mm{E}) \right] (\mm{I} - \tilde{\mm{C}})^{-1} \\
				&= \left[ (\hat{\mm{C}} + \hat{\mm{E}}) - \hat{\mm{C}}( \mm{I} - (\mm{I}-\mm{C})^{-1}\mm{E}) \right] (\mm{I} - \tilde{\mm{C}})^{-1} \\
				&= \left[ \hat{\mm{E}} + \hat{\mm{C}}(\mm{I}-\mm{C})^{-1}\mm{E} \right] (\mm{I} - \tilde{\mm{C}})^{-1} \\
		\end{align*}

		Now, we notice that
		\begin{align*}
			\left\| (\mm{I}-\tilde{\mm{C}})^{-1} \right\| 	&= \left\| \sum_{k=0}^\infty \tilde{\mm{C}}^{k} \right\| 
													\le \sum_{k=0}^\infty \left\| \tilde{\mm{C}}^{k} \right\| 
													\le \sum_{k=0}^\infty (1-r)^k 
													= \frac{1}{r}
		\end{align*}

		Because money is never created or destroyed, we have 
		\[
			\left\| \hat{\mm{C}}(\mm{I}-\mm{C})^{-1} \vv{v} \right\|_1 = \left\| \vv{v} \right\|_1
		\]

		Thus we have
		\begin{align*}
			\left\| \hat{\tilde{\mm{C}}}(\mm{I} - \tilde{\mm{C}})^{-1} - \hat{\mm{C}}(\mm{I} - \mm{C})^{-1} \right\| 
			&= \left\| \left[ \hat{\mm{E}} + \hat{\mm{C}}(\mm{I}-\mm{C})^{-1}\mm{E} \right] (\mm{I} - \tilde{\mm{C}})^{-1} \right\| \\
			&\le \left\| \hat{\mm{E}} \right\| \cdot \left\| (\mm{I} - \tilde{\mm{C}})^{-1} \right\| + \left\| \hat{\mm{C}}(\mm{I}-\mm{C})^{-1}\right\|\cdot \left\| \mm{E} \right\| \cdot \left\|(\mm{I} - \tilde{\mm{C}})^{-1} \right\| \\
		\end{align*}
		
		Thus we immediately get the bound
		\[
			\left\| \hat{\tilde{\mm{C}}}(\mm{I} - \tilde{\mm{C}})^{-1} - \hat{\mm{C}}(\mm{I} - \mm{C})^{-1} \right\| \le \frac{\epsilon}{r}
		\]

		Since $\left\| \hat{\mm{C}}(\mm{I} - \mm{C})^{-1} \right\| \le 1$, we also have the trivial bound 
		\[
			\left\| \hat{\tilde{\mm{C}}}(\mm{I} - \tilde{\mm{C}})^{-1} - \hat{\mm{C}}(\mm{I} - \mm{C})^{-1} \right\| \le 2
		\]
	\end{proof}

	The multiplicative bound of $\nicefrac{\eps}{r}$ in Theorem \ref{thm:cyclicsensitivity} is much weaker than the bound of $\eps$ in the acyclic case (Corollary \ref{cor:acyclicsensitivity}), 
	as the minimum self-holdings approaches $0$, this difference tends towards infinity.  This discrepancy is not a limitation of our proof, 
	but arises as an artifact of the effect that holdings cycles can have on the equity (and market) valuations of the institutions in the network.
	In Theorem \ref{thm:cycliclowerbound} we show that there exist networks where changing a single institution's holdings by $\eps$ results in a change 
	of $\frac{\eps}{2r + \frac{1-r}{2} \eps} \|\mm{D}\vv{p}\|$ in one of the institution's market values,
	where $r$ is the minimum self-holdings in the network.

	\begin{theorem}
		\label{thm:cycliclowerbound}
		There exist networks where 
		$\|\mm{C} - \tilde{\mm{C}}\| \le \eps$, and 
		$\|\vv{v} - \tilde{\vv{v}} \| \ge \frac{\eps}{2r + \frac{1-r}{2}} \left\| \mm{D} \vv{p} \right\|$, where 
		\[
			\vv{v} = \hat{\mm{C}} (\mm{I} - \mm{C})^{-1} \mm{D} \vv{p}
		\]
		and $r = \min_{i}( \tilde{\hat{C}}_{ii},\hat{C}_{ii} )$ is the minimum ``self-holdings'' of the financial institutions.

	\end{theorem}
	
	\begin{proof}
		We exhibit an initial network in Figure \ref{fig:configuration1} and its perturbation in Figure \ref{fig:configuration2}.

		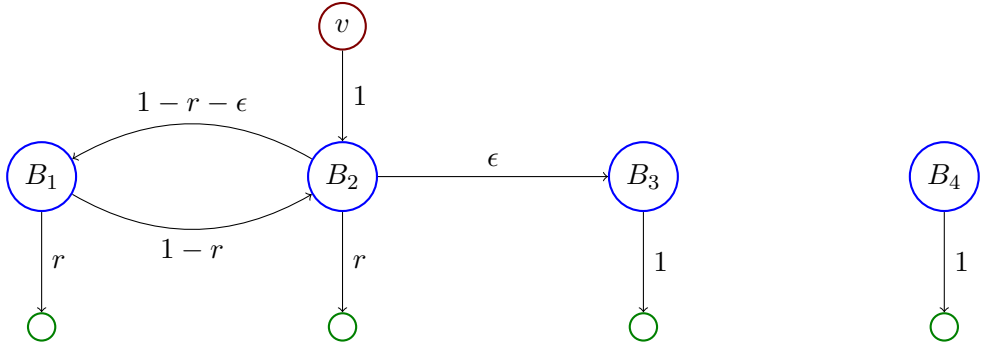
\begin{figure}
			\begin{center}
				\tikzstyle{asset}=[circle,thick,draw=red!50!black]
				\tikzstyle{bank}=[circle,thick,draw=blue]
				\tikzstyle{shadow}=[circle,thick,draw=green!50!black]
				\begin{tikzpicture}

					\node (A) [bank] at (-2,0) {$B_1$};
					\node (B) [bank] at ([xshift=4cm]A) {$B_2$};
					\node (C) [bank] at ([xshift=4cm]B) {$B_3$};
					\node (D) [bank] at ([xshift=4cm]C) {$B_4$};

					\node (ASSET2) [asset] at ([yshift=2cm]B) {$v$};

					\node (AA) [shadow] at ([yshift=-2cm]A) {};
					\node (BB) [shadow] at ([yshift=-2cm]B) {};
					\node (CC) [shadow] at ([yshift=-2cm]C) {};
					\node (DD) [shadow] at ([yshift=-2cm]D) {};

					\draw [->,bend right] (A) to node [below] {$1-r$} (B);
					\draw [->,bend right] (B) to node [above] {$1-r-\epsilon$} (A);
					\draw [->] (B) to node [above] {$\epsilon$} (C);

					\draw [->] (A) to node [right] {$r$} (AA);
					\draw [->] (B) to node [right] {$r$} (BB);
					\draw [->] (C) to node [right] {$1$} (CC);
					\draw [->] (D) to node [right] {$1$} (DD);
				
					\draw [->] (ASSET2) to node [right] {$1$} (B);

				\end{tikzpicture}
				\caption{The initial configuration, banks are in blue, external shareholders are in green, and the asset is shown in red.\label{fig:configuration1}}
			\end{center}
		\end{figure}
	
	In Figure \ref{fig:configuration1}, the equity values for the banks satisfy
	\begin{align*}
		B_1 &= (1-r-\epsilon)B_2 \\
		B_2 &= v + (1-r)B_1 \\
		B_3 &= \epsilon B_2 \\
		B_4 &= 0 
	\end{align*}

	So 
	\begin{align*}
		B_2 &= v + (1-r)B_1 \\
			&= v + (1-r)(1-r-\epsilon)B_2 \\
			&= v + (1-r-\epsilon -r + r^2 +r\epsilon) B_2 \\
	\end{align*}
	Rearranging gives
	\begin{align*}
		B_2	&= v + (1-r-\epsilon -r + r^2 +r\epsilon) B_2 \\
			&\Downarrow \\
		v &= (2r+\epsilon-r^2-r\epsilon)B_2  \\
			&\Downarrow \\
		B_2 &= \frac{v}{2r+\epsilon - r^2 -r\epsilon} \\
			&= \frac{v}{r(2-r) + (1-r)\epsilon}\\
			&\ge \frac{1}{2} \left( \frac{v}{r + \frac{1-r}{2}\epsilon} \right)\\
	\end{align*}

	Thus the market valuation of $B_3$ is $\epsilon B_2$ which is at least $\frac{1}{2}\left( \frac{\eps v}{r + \frac{1-r}{2}\epsilon} \right)$.

	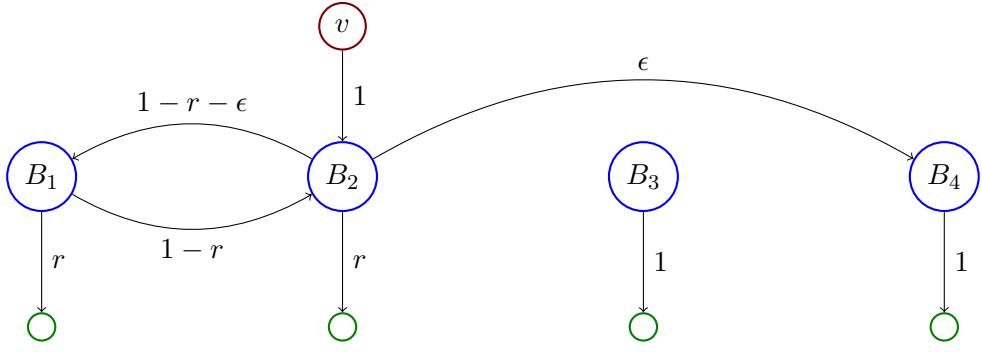
\begin{figure}
		\begin{center}
			\tikzstyle{asset}=[circle,thick,draw=red!50!black]
			\tikzstyle{bank}=[circle,thick,draw=blue]
			\tikzstyle{shadow}=[circle,thick,draw=green!50!black]
			\begin{tikzpicture}

				\node (A) [bank] at (-2,0) {$B_1$};
				\node (B) [bank] at ([xshift=4cm]A) {$B_2$};
				\node (C) [bank] at ([xshift=4cm]B) {$B_3$};
				\node (D) [bank] at ([xshift=4cm]C) {$B_4$};

				\node (ASSET2) [asset] at ([yshift=2cm]B) {$v$};

				\node (AA) [shadow] at ([yshift=-2cm]A) {};
				\node (BB) [shadow] at ([yshift=-2cm]B) {};
				\node (CC) [shadow] at ([yshift=-2cm]C) {};
				\node (DD) [shadow] at ([yshift=-2cm]D) {};

				\draw [->,bend right] (A) to node [below] {$1-r$} (B);
				\draw [->,bend right] (B) to node [above] {$1-r-\epsilon$} (A);
				\draw [->,bend left] (B) to node [above] {$\epsilon$} (D);

				\draw [->] (A) to node [right] {$r$} (AA);
				\draw [->] (B) to node [right] {$r$} (BB);
				\draw [->] (C) to node [right] {$1$} (CC);
				\draw [->] (D) to node [right] {$1$} (DD);
			
				\draw [->] (ASSET2) to node [right] {$1$} (B);

			\end{tikzpicture}
			\caption{The perturbed configuration, where one link of weight $\epsilon$ has been moved from $B_3$ to $B_4$.\label{fig:configuration2}}
		\end{center}
	\end{figure}

	If the link from $B_2 \rightarrow B_3$ were moved to $B_2 \rightarrow B_4$ (as in Figure \ref{fig:configuration2}) then 
	$B_3$'s value drops to zero and $B_4$'s value increases to $\epsilon B_2$.

	Thus the change in $\ell_1$ norm of the market valuations between the two situations is at least 

	\[
		\frac{\epsilon v}{r + \frac{1-r}{2}\epsilon}
	\]
	
	Writing this in matrix notation, we have

	\[
		\mm{C} = \left[ \begin{array}{cccc} 0 & 1-r-\epsilon & 0 & 0 \\ 1-r & 0 & 0 & 0 \\ 0 & \epsilon & 0 & 0 \\ 0 & 0 & 0 & 0 \end{array} \right]
		\qquad
		\tilde{\mm{C}} = \left[ \begin{array}{cccc} 0 & 1-r-\epsilon & 0 & 0 \\ 1-r & 0 & 0 & 0 \\ 0 & 0 & 0 & 0 \\ 0 & \epsilon & 0 & 0 \end{array} \right]
	\]
	\[
		\hat{\mm{C}} = \hat{\tilde{\mm{C}}}= \left[ \begin{array}{cccc} r & 0 & 0 & 0 \\ 0 & r & 0 & 0 \\ 0 & 0 & 1 & 0 \\ 0 & 0 & 0 & 1 \end{array} \right]
	\]
	\end{proof}

	Note that increasing one link by $\eps$ and decreasing another by $\eps$ is actually a change in $2\eps$ in the $\|\mm{C} - \hat{\mm{C}}\|$.
	Letting $\eps' = \frac{\eps}{2}$, we have a change of $\eps$ in $\|\mm{C} - \hat{\mm{C}}\|$ yields a change 
	of at most $\frac{\eps}{2r + \frac{1-r}{2} \eps}$.  Notice that as $r \rightarrow 0$, this is approaches $2$, 
	\ie the resulting valuation is \emph{as far as possible} in terms of the $\ell_1$ norm of the market values of 
	the institutions.

	The type of equity amplification necessary for the lower bound in Theorem \ref{thm:cycliclowerbound} cannot happen in an acyclic network (see Corollary \ref{cor:acyclicsensitivity}).

	Now, we compare the upper and lower bounds on sensitivity.  
	We have an upper bound of $\max( 2, \frac{\eps}{r} )$ and a lower bound of $\frac{\eps}{2r + \frac{1-r}{2} \eps}$.
	The lower bound implies the following.
	\begin{itemize}
		\item
			For any $\eps,\delta > 0$, there exists an $r$ with $0 < r < 1$ and a network with minimum reserve $r$ such that a change in edge weights of $\eps$ can lead to a change in market valuations of at least $(1-\delta) 2$.
			Thus the upper bound cannot be decreased below $2$.
		\item
			For any $r,\delta$, with $0 < r < 1$, and $0 < \delta$ there exists an $\eps$ with $0 < \eps < 1$ and a network with minimum reserve $r$ such that a change in edge weights of $\eps$ can lead to a change in market valuations of at least $(1-\delta) \frac{\eps}{2r}$, 
			thus the upper bound cannot be decreased below $\frac{\eps}{2r}$.
	\end{itemize}

	Together, these show that the upper bound cannot be decreased below $\max( 2, \frac{\eps}{2r} )$, so our upper bound of $\max( 2, \frac{\eps}{r} )$ is essentially tight.	
	Another interpretation of the lower bound is that for any $n > 0$, there exists an $\eps,r > 0$ and a network such that an $\eps$ change in one edge weight can cause a multiplicative change of $n \eps$ in market malues of the institutions.

%% file: complexity.tex
\section{Bank failures}

	\subsection{Losses caused by failure}	
		The model of \cite{EGJ14} includes a notion of ``failure,'' whereby institutions whose market value drops below a certain critical 
		threshold suffer a further (discontinuous) loss in market value.  These discontinuous penalties capture the notion that if an institution cannot 
		pay its operating costs, it may see a further drop in revenues.  Similarly, if confidence in the institution is shaken, and its debt rating is downgraded, 
		it may see spike in the cost of capital, and hence see a further drop in value.

		These discontinuous penalties are operationalized by a threshold value $\munderbar{v}_i$, such that if institution $i$'s market value, $v_i$, drops 
		below $\munderbar{v}_i$ then it incurs a failure cost and its market value drops by an additional $\beta_i(\vv{p})$.

		Defining $I_{v_i < \munderbar{v}_i}$ to be the indicator variable which is 1 if $v_i < \munderbar{v}_i$ and 0 otherwise, and 
		$b_i(\vv{v},\vv{p}) = \beta_i(\vv{p}) I_{v_i < \munderbar{v}_i}$, the market value of the institutions satisfies the equation
	
		\begin{equation}
			\label{eqn:marketvaluediscontinuous}
			\vv{v} = \hat{\mm{C}}(\mm{I} - \mm{C})^{-1} (\mm{D} \vv{p} - \vv{b}(\vv{v}))
		\end{equation}
		
		Compare Equation \ref{eqn:marketvaluediscontinuous} to the linear system given by Equation \ref{eqn:marketvaluation}.
		The introduction of non-linear terms into the model adds significant complexity to the dynamics of the system, and 
		can lead to failure cascades.  One of the primary goals of \cite{EGJ14} was to characterize 
		what network features affect the likelihood and severity of failure cascades.		

	\subsection{Overview of our complexity results}

		Suppose a regulator has complete information about a financial network, including all cross-holdings and the prices of all underlying assets.
		Further, suppose that at equilibrium, this network has no failures.  If the regulator believes that tomorrow, asset prices may drop by some fixed amount $d$, 
		(\ie the sum total of asset prices may drop by $d$, but the exact drop of each asset price is unknown), what is the maximum number failures that could 
		occur as a result of this drop?  In other words, tomorrow, when the new equilibrium is calculated based on the new, lower, asset prices, what is 
		the maximum number of banks that could have failed?

		Our primary result is that the introduction of discontinuous failure costs increases the computational complexity 
		of calculating basic network dynamics.  If the complete network, as well as the price of all underlying assets is 
		known, then the number of failures can be computed efficiently \cite[Section 3.2.3]{EGJ14}.
		If, on the other hand, there is some uncertainty in the prices (or future prices) of the underlying assets, 
		calculating the maximum number of failures that could occur in the network is computationally intractable.

		Thus we address the following question: given a stable network (where no banks have failed), if the total prices of the assets 
		drop by some small amount, what is the maximum number of failures that occur at equilibrium?  
		One potential complication is that in networks with discontinuous failure penalties, 
		there may not be a unique equilibrium, and hence the number of failures may not be uniquely defined.  
		To address this, for any fixed set of asset prices, we use standard practice and only consider the ``best-case'' equilibrium (the one with fewest failures).
		On the other hand, we consider the worst drop in asset prices (\ie the drop in asset prices that causes the most failures in its best-case equilibrium).
		Thus the ``maximum'' number of failures means the maximum over all bounded drops in asset values, of the minimum number failures 
		that could occur at this each fixed drop in asset values (\ie the minimum over all equilibria at these new, lower values).
		We discuss the issue of multiple equilibria in more detail in Section \ref{sec:multipleequilibria}.

		In real-world financial networks, regulators, as well as the institutions themselves routinely perform ``stress-tests'' to assess the robustness 
		of the network to financial shocks.  Our results indicate that there are financial networks where no (computationally feasible) stress test, will 
		ever be able to even approximate the maximum number of failures that could occur from some small shock to the system.  
	
		Note, however, that our results are not probabilistic, in the sense that we do not model the \emph{probability} of a specific drop in asset prices, instead 
		we just impose a bound on the magnitude of the total drop in asset values.
		Without assigning probability distributions to the asset prices, we cannot assess the probability that such a cascade could occur, only that it is feasible.
		Thus our question is whether there exists a small specific drop in asset prices that could cause a huge failure cascade -- whether this drop is \emph{probable} 
		is outside the scope of our model.
		Thus, if a stress-test could accurately model the probability distribution of asset prices, it may be able to assess the \emph{likely} number of failures 
		caused by a shock to the system, but any computationally feasible stress test will not be able to identify whether a small (but possibly unlikely) shock to the system 
		could cause a far greater failure cascade.


	\subsection{Multiple equilibria}
		\label{sec:multipleequilibria}
		
		When discontinuous failure penalties ($\vv{b}(\vv{v})$) are introduced into the system, then there may be multiple equilibrium values for the institutions in the system, \ie 
		the market values may not be uniquely defined \cite[Section 2.6 and Appendix A.7]{EGJ14}.  As noted \cite{EGJ14}, there are two distinctly different ways a network 
		can have multiple equilibria.  The first type falls into the standard theory of bank runs \cite{DD83}, and this type of multiplicity can occur even in an acyclic network.
		For example suppose there is a network consisting of single institution holding a single asset of value $p$.  Let $v$ denote the value of the institution and $\underbar{v}$ denote 
		its failure threshold, and $\beta$ its failure penalty.  If $p > \underbar{v} > p - \beta$, then a valuation of $v = p$, and $v = p - \beta$ are both consistent with Equation \ref{eqn:marketvaluediscontinuous}.
		The second type of multiplicity is caused by cycles in the network.  See Figure \ref{fig:multipleequilibria} for a simple (cyclic) network that has multiple equilibria.
		
		In a network with multiple equilibria, we can talk about the ``best-case'' equilibrium (the one with the fewest failures) and the ``worst-case'' equilibrium (the one with the most failures).
		As in \cite{EGJ14}, we focus on the best-case equilibrium, thus when we refer to the ``the number of failures'' we mean the number of failures in the best-case equilibrium.

		\begin{figure}
			\begin{center}
				\tikzstyle{asset}=[circle,thick,draw=red!50!black]
				\tikzstyle{bank}=[circle,thick,draw=blue]
				\tikzstyle{shadow}=[circle,thick,draw=green!50!black]
				\begin{tikzpicture}
					\node (A1) [asset] at (0,2) {$1$};
					\node (A2) [asset] at (2,2) {$1$};
					\node (B1) [bank] at (0,0) {$B_1$};
					\node (B2) [bank] at (2,0) {$B_2$};
					\node (S1) [shadow] at ([yshift=-2cm]B1) {};
					\node (S2) [shadow] at ([yshift=-2cm]B2) {};

					\draw [->] (A1) to node [left] {1} (B1);
					\draw [->] (A2) to node [right] {1} (B2);
					\draw [->] (B1) to node [left] {$\half$} (S1);
					\draw [->] (B2) to node [right] {$\half$} (S2);
			
					\draw [->,bend left] (B1) to node [above] {$\nicefrac{1}{2}$} (B2);
					\draw [->,bend left] (B2) to node [below] {$\nicefrac{1}{2}$} (B1);
				\end{tikzpicture}
			\end{center}
			\caption[A simple network with multiple equilibria.]{%
				Banks are in blue, external shareholders are in green, and the asset is shown in red.
				In this example $\mm{D}\vv{p} = \left[ \begin{array}{cc} 1 & 0 \\ 0 & 1 \end{array} \right]$, and $\mm{C} = \left[ \begin{array}{cc} 0 & \half \\ \half & 0 \end{array} \right]$.
				If $\munderbar{v}_i = 2$, and $b_i = 1$, then two equilibria are: $\vv{v} = \left[ \begin{array}{c} 1 \\ 1 \end{array} \right]$,
				and $\vv{v} = \left[ \begin{array}{c} 0 \\ 0 \end{array} \right]$

			\label{fig:multipleequilibria}}
		\end{figure}

	\subsection{The Balanced Complete Bipartite Subgraph (BCBS) problem}

		Our hardness result is based on the hardness of finding a maximum balanced clique in a bipartite graph.
		This is known as Balanced Complete Bipartite Subgraph (BCBS) problem.
		
		\begin{definition}[BCBS]
			Given a bipartite graph $G = (V_1,V_2,E)$ with $|V_1| = |V_2| = n$, the \emph{Balanced Complete Bipartite Subgraph} (BCBS) 
			problem is to find the largest integer $K$ such that there exists sets $C_1 \subset V_1$ and $C_2 \subset V_2$ 
			with the properties that $|C_1| = |C_2| = K$, and the induced graph on $C_1 \cup C_2$ is a complete bipartite subgraph of $G$.
		\end{definition}

		The BCBS problem is known to be NP-hard \cite{GJ79,J87}.
		This provides strong evidence that there is no scalable algorithm 
		that it that can find the size of the maximum balanced clique in a bipartite graph.
		
		In fact, there is significant evidence that even \emph{approximating} the size of the largest 
		balanced clique is hard.
		Feige showed that for some $\delta > 0$ it is Random 3-SAT hard to
		approximate BCBS to within a factor of $n^\delta$ \cite[Theorem 3]{F02}.

		Feige and Kogan showed that if BCBS can be approximated to within a factor of $2^{(\log n)^\delta}$ 
		for every $\delta > 0$ then 3-SAT can be solved in time $2^{n^{3/4}+\eps}$ for every $\eps > 0$ 
		\cite[Theorem 1.3]{FK04}.

		Feige and Kogan go on to conjecture that for some $\delta > 0$, there is no polynomial-time algorithm to approximate 
		BCBS to within a factor of $n^\delta$ \cite[Conjecture 1.1]{FK04}.

		Our primary hardness result shows that there are networks where calculating the maximum number of bank failures 
		that can result from a small shock is equivalent to solving the BCBS problem.
		Thus there are financial networks where even approximately estimating the number of bank failures that can result from an 
		arbitrarily small drop in asset prices is a computationally intractable problem.

	\subsection{The complexity of calculating the maximum number of failures}

		In this section, we give our main result 
		concerning the computational complexity of estimating the maximum number of failures 
		that can occur given a small drop in values of the underlying assets.
	
		In particular, this hardness result applies finding the number of failures that occur at equilibrium 
		when the institutional cross-holdings are fixed, 
		and completely known, but there is some small uncertainty in the prices of the underlying assets.
		
		In the situation that the cross-holdings and the assets are fixed, it is straightforward to calculate the 
		market values of the institutions at equilibrium \cite[Section 3.2.3]{EGJ14}.  
		Note that the network we construct is acyclic.

		\begin{theorem}
			\label{thm:hardness}
			For every bipartite graph $G$ on $2n$ nodes, and every $\eps > 0$, there is a financial network with $\Omega(n)$ 
			institutions, and a $d > 0$ such that computing the maximum number of institutions that could fail 
			following a shock of $d \eps$ in asset prices is as hard as solving the BCBS problem in $G$.
		\end{theorem}

		\begin{proof}
			Let $\ell > 0$ be any integer. 

			Our starting point is the following hardness result for the BCBS problem. 
			Given an $n \times n$ balanced bipartite graph $G$, it is hard to decide whether 
			the largest balanced bipartite clique size in $G$ is at least $K \times K$ or at most $K/g \times K/g$ 
			for some gap function $g$. For instance, $g = 2^{(\log n)^\delta}$ under the assumption 
			that 3-SAT $\not \in DTIME\inparen{2^{n^{3/4+\eps}}}$ for some $\eps > 0$.

			Given an $n \times n$ balanced bipartite graph $G$, 
			we will construct a financial network with $(2+\ell)n$ institutions 
			such that if $G$ has a balanced bipartite subgraph of size $K$, 
			then a drop in asset prices by $K\eps$ can cause at least $(2+\ell)K$ failures.
			On the other hand, if the largest balanced bipartite subgraph of $G$ is of 
			size $\frac{K}{g}$, a drop in asset prices of $K\eps$ can cause at most $K + \frac{K}{g}(\ell + 1)$ failures.

			This shows that estimating the maximum number of failures induced by a fixed 
			drop in asset prices is at least as hard as estimating the size of the maximum 
			balanced bipartite clique.
			Without loss of generality, we assume that every vertex in $G$ has degree at least $K$.

			Let $D$ denote the maximum degree of any vertex in $G$.
			Let $0 < \eps < 1$ be an arbitrary parameter, and let $0 < r < 1$ denote 
			the minimum amount of self-holdings of the institutions in the network we are 
			constructing.  (The reduction will hold for any choices of $0 < \eps,r < 1$.)

			For each node in the graph $G$, we will associate a financial institution.
			Let institutions $1,\ldots,n$ correspond to the left-hand nodes of $G$, 
			and institutions $n+1,\ldots,2n$ correspond to the right-hand nodes of $G$.

			We will also generate $n$ underlying assets, labelled $a_i',\ldots,a_n'$ 
			and institution $i$ will complete own asset $a_i'$ for $i=1,\ldots,n$.
			Institutions $n+1,\ldots,2n$ will own none of the underlying assets.
			All assets will initially be valued at $1$.
			Thus
			\[
				\mm{D} = \left[ \begin{array}{cccccc} 
				1 & & & & &\\
				  & \ddots & & & & \\
				  & & 1 & & & \\
				 & & & 0 & & \\
				 & & & & \ddots & \\
				 & & & &   &0 \\
				\end{array} \right]
			\]

			Define $N = \frac{D}{1-r}$ (recall $D$ is the maximum degree of $G$, and $r$ is an arbitrary parameter that will determine the \emph{integration} of the resulting 
			financial network).  We will use $\Gamma(j)$ to denote the neighbors of vertex $j$ in $G$.
			Notice that our definition of $N$ ensures that 
			\[
				1-r = \frac{D}{N} \ge \frac{|\Gamma(j)|}{N}
			\]
			for all $j = 1,\ldots,2n$.
			This means that if institution $j$ sells an equal $\frac{1}{N}$ stake in itself to all of its neighbors, 
			it will be left with at least an $r$ fraction of self ownership.  To operationalize this, we define
			\[
				c_{ij} = \left\{ \begin{array}{cl} \frac{1}{N} &\mbox{ if $i > j$ and $(i,j)$ an edge of $G$ } \\ 0 &\mbox{ otherwise } \end{array} \right.
			\]

			For $1=1,\ldots,n$ let $\munderbar{v}_i = 1-\frac{|\Gamma(i)|}{N} - \epsilon$, so if institution $i$'s asset drops in value by $\epsilon$ then institution $i$ will fail.
			Let the failure penalty $\beta_i = 1-\frac{|\Gamma(i)|}{N} - \epsilon$ for $i=1,\ldots,n$.  Thus if asset $i$ drops in price by $\epsilon$, institution $i$ fails and its value immediately drops to $0$.
			For $i=n+1,\ldots,2n$ let $\munderbar{v}_i = \frac{|\Gamma(i)|-d}{N}$.  Notice that if all assets are initially valued at $1$ then
			\[
				v_i = \left\{ \begin{array}{cl} 1 - \frac{|\Gamma(i)|}{N}  &\mbox{ if $1 \le i \le n$ } \\ \\ \frac{|\Gamma(i)|}{N} &\mbox{ if $n < i \le 2n$ } \end{array} \right.
			\]

			This financial network has the following properties:

			\begin{enumerate}
				\item
					If $j > n$ and $d$ of bank $j$'s neighbors fail, \ie $d$ of the assets $a_i$ ($i \in \Gamma(j)$) drop in value by $\epsilon$, 
					then bank $j$ will fail.
				\item
					If $j > n$ and the total drop in value of bank $j$'s neighbors is less than $\frac{d}{N}$ then bank $j$ will not fail.
			\end{enumerate}

			Now, we examine the properties of this system when the assets are allowed to drop in price by total amount $d\eps$.
			A drop of $d \eps$ can always cause $d$ institutions on the left-hand side of the network to fail, simply 
			by dropping the value of each of their assets by $\eps$.
		
			What happens if the price drop is not concentrated among exactly $d$ assets?
			Let $t$ denote the number of assets that drop in value by at least $\eps$.
			If $t < d$, then at least $t \eps$ from the ``shock budget'' of $d\eps$ 
			was used to lower the price of these $t$ assets, which leaves a budget of $(d-t)\eps$ remaining.
			Now, consider how much this drop can affect one of a right-hand institutions, $j$.
			Even if this drop is concentrated entirely among the left-hand neighbors of $j$, 
			the drop $j$ feels is at most
			\[	
				\frac{1}{N} \left( t + (d-t)\epsilon \right) < \frac{1}{N} \left( t + (d-t) \right) = \frac{d}{N}
			\]
			Thus institution $j$ cannot fail, since the failure of a right-hand institution requires a drop in value of at least $\frac{d}{N}$.
			This means that in this case exactly $t < d$ institutions fail.
			Since we are interested in the \emph{maximum} number of failures that can arise from a drop in asset value of $d \eps$, 
			we can, without loss of generality, assume that \emph{exactly} $d$ assets drop in value by $\eps$, 
			causing \emph{exactly} $d$ failures among the left-hand institutions.

			Now, suppose there is a biclique of size $K$ in $G$.
			If $d = K$, then causing $d$ failures among the left-hand members of this biclique 
			will cause $d=K$ failures among the right-hand members.

			On the other hand, suppose the largest biclique is of size $\frac{K}{g}$.
			If the failure of $d$ left-hand institutions causes the failure of $P$ 
			right-hand institutions, then each of the $P$ failed institutions on the right 
			must be connected to each of the $d$ failed institutions on the left.
			Thus there must be a biclique of size $\min(d,P) = P$.

			Thus in the ``yes'' case ($G$ has a biclique of size $K$) we can cause at least $K$ right hand failures 
			with a failure budget of $K \eps$.  In the ``no'' case (the largest biclique in $G$ is of size $\nicefrac{K}{g}$) 
			the maximum number of failures of right-hand institutions is bounded by $\frac{K}{g}$.

			Now, to amplify this discrepancy, we add a chain of $\ell$ institutions connected to each right hand institution.
			Thus for every right-hand institution, $b_i$, it will have a chain of institutions $b_i^{(1)},\ldots,b_i^{(\ell)}$ where 
			$b_i^{(j)}$ owns a $1-r$ fraction of $b_i^{(j-1)}$ and has no other holdings.  Thus if $b_i$ fails, 
			then $b_i^{(1)}$ through $b_i^{(\ell)}$ fail as well (see Figure \ref{fig:hardnetwork}).

			This new network has $(\ell+2)n$ banks, and has the following properties.
			Given a failure budget of $K \eps$,
			if the largest balanced clique in $G$ is a $K \times K$, a drop in value by $K \epsilon$ causes at most $(2+\ell)K$,
			but if the largest biclique is of size $\frac{K}{g}$, then a drop in asset values 
			of at most $K \eps$ can cause at most $K + (\ell+1)\frac{K}{g} = (g+\ell+1)\frac{K}{g}$ failures.

			Note that when the gap, $g = n^\delta$, choosing $\ell = poly(n)$, we obtain a gap of $((\ell+2)n)^{\delta'}$ for some 
			$\delta' < \delta$.

		\end{proof}

		\begin{center}
			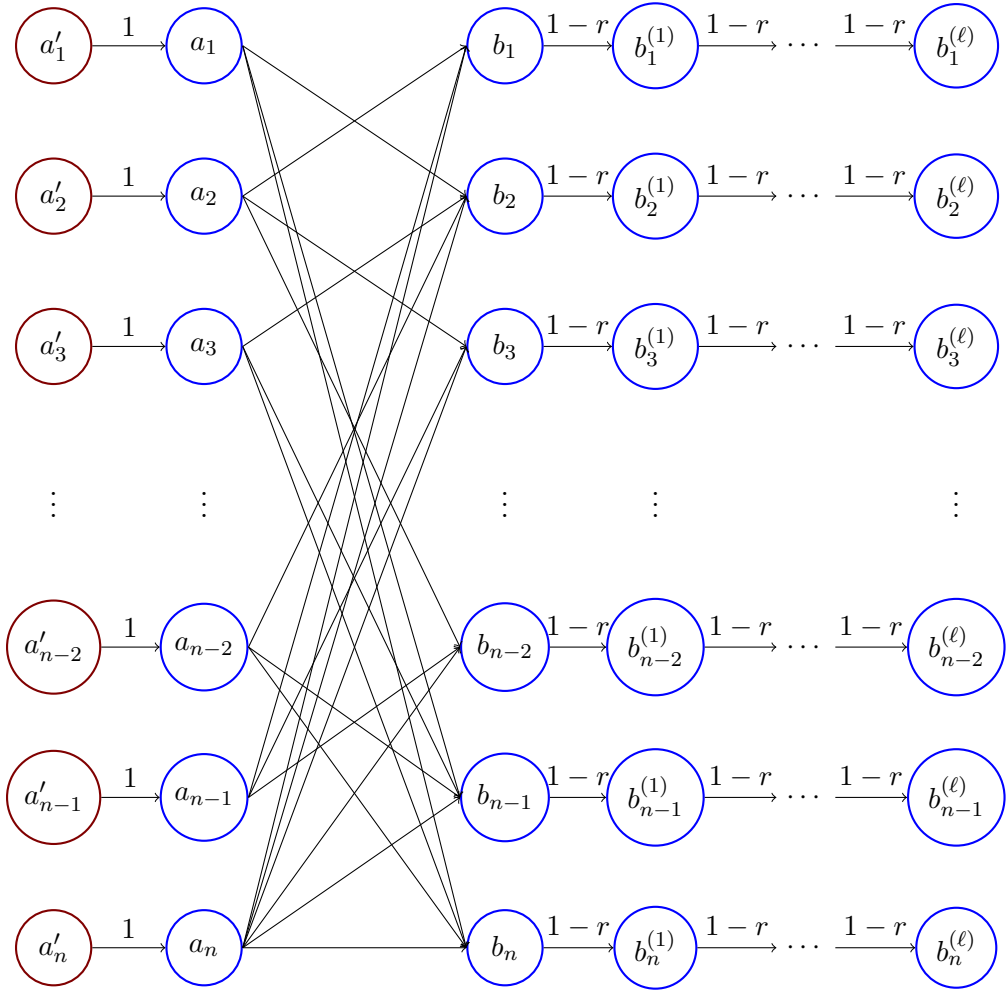
\begin{figure}
			\tikzstyle{asset}=[circle,thick,draw=red!50!black,minimum size=1cm]
			\tikzstyle{bank}=[circle,thick,draw=blue,minimum size=1cm]
			\tikzstyle{shadow}=[circle,thick,draw=green!50!black]
			\begin{tikzpicture}
		
				\coordinate (TOP) at (2,8);
				\coordinate (BOT) at (2,-8);
		
				\pgfmathtruncatemacro{\LAST}{7} 
				\pgfmathtruncatemacro{\MID}{(\LAST+1)/2}						
				\foreach \xi in {1,...,\LAST} {
					\pgfmathsetmacro{\x}{\xi/(\LAST+1)}

					\ifthenelse{\xi=\MID}{
						\node (AA\xi) at ([xshift=0]$(TOP)!\x!(BOT)$) {$\vdots$};
						\node (A\xi) at ([xshift=2cm]$(TOP)!\x!(BOT)$) {$\vdots$};

						\node (BB\xi) at ([xshift=6cm]$(TOP)!\x!(BOT)$) {$\vdots$};
						\node (B\xi) at ([xshift=8cm]$(TOP)!\x!(BOT)$) {$\vdots$};
						\node (B\xi) at ([xshift=12cm]$(TOP)!\x!(BOT)$) {$\vdots$};
					}{
						\ifthenelse{\xi=\LAST}{
							\node (AA\xi) at ([xshift=0]$(TOP)!\x!(BOT)$) [asset] {$a^\prime_n$};
							\node (A\xi) at ([xshift=2cm]$(TOP)!\x!(BOT)$) [bank] {$a_n$};
							\draw [->] (AA\xi) to node [above] {$1$} (A\xi);	

							\node (BB\xi) at ([xshift=6cm]$(TOP)!\x!(BOT)$) [bank] {$b_n$};
							\node (B\xi) at ([xshift=8cm]$(TOP)!\x!(BOT)$) [bank] {$b^{(1)}_n$};
							\node (BBBB\xi) at ([xshift=12cm]$(TOP)!\x!(BOT)$) [bank] {$b^{(\ell)}_{n}$};
							\node (BBB\xi) at ([xshift=10cm]$(TOP)!\x!(BOT)$)  {$\cdots$};
							\draw [->] (BB\xi) to node [above] {$1-r$} (B\xi);	
							\draw [->] (B\xi) to node [above] {$1-r$} (BBB\xi);	
							\draw [->] (BBB\xi) to node [above] {$1-r$} (BBBB\xi);	
						}{
							\ifthenelse{\xi>\MID}{
								\pgfmathtruncatemacro{\C}{\LAST-\xi}
								\def\A{n-\C}	
							}{
								\def\A{\xi}
							}
							\node (AA\xi) at ([xshift=0]$(TOP)!\x!(BOT)$) [asset] {$a^\prime_{\A}$};
							\node (A\xi) at ([xshift=2cm]$(TOP)!\x!(BOT)$) [bank] {$a_{\A}$};
							\draw [->] (AA\xi) to node [above] {$1$} (A\xi);	

							\node (BB\xi) at ([xshift=6cm]$(TOP)!\x!(BOT)$) [bank] {$b_{\A}$};
							\node (B\xi) at ([xshift=8cm]$(TOP)!\x!(BOT)$) [bank] {$b^{(1)}_{\A}$};
							\node (BBB\xi) at ([xshift=10cm]$(TOP)!\x!(BOT)$)  {$\cdots$};
							\node (BBBB\xi) at ([xshift=12cm]$(TOP)!\x!(BOT)$) [bank] {$b^{(\ell)}_{\A}$};
							\draw [->] (BB\xi) to node [above] {$1-r$} (B\xi);	
							\draw [->] (BBB\xi) to node [above] {$1-r$} (BBBB\xi);	
							\draw [->] (B\xi) to node [above] {$1-r$} (BBB\xi);	
						}
					}		
				}

				\foreach \x in {1,...,\LAST} {
					\foreach \a/\b in {3/4,3/2,3/6} {
						\pgfmathtruncatemacro{\C}{mod(\a*\x+\b,6)+1}
						\ifthenelse{\x=\MID}{
							\pgfmathtruncatemacro{\y}{\LAST}
						}{
							\pgfmathtruncatemacro{\y}{\x}
						}
						\ifthenelse{\C=\MID}{
							\pgfmathtruncatemacro{\C}{\LAST}
						}{}
						\draw [->] (A\y.east) to (BB\C.west);
					}
				}

			\end{tikzpicture}
			\caption{There are $n$ assets, $\{a_i^\prime\}$, shown in {\color{red!50!black} red}.  Each asset $a_i^\prime$ is fully owned by institution $a_i$.  
				Institutions $\{a_i\}$, $\{b_i\}$ correspond to a hard instance of balanced bipartite clique.
				Institutions $\left\{ b_i^{(j)} \right\}$ serve to amplify the failures that occur in the first level. \label{fig:hardnetwork}}
			\end{figure}
		\end{center}	

		Applying a result of \cite{FK04}, we obtain the following Corollary.

		\begin{corollary}
			If $3$-SAT $\not \in DTIME\inparen{ 2^{n^{3/4+\eps}}}$ for some $\eps > 0$, then 
			there exists a $\delta > 0$ such that there is no polynomial time algorithm that can calculate the maximum number of failures 
			in a financial network caused by a drop in asset prices of $d \eps$ to within a 
			factor of $2^{(\log n)^{\delta'}}$ for some $\delta' > 0$.
		\end{corollary}

%% file: conclusion.tex
\section{Conclusion}

	This work highlights two distinct sources of instability in financial networks, 
	instability arising from fluctuations in cross-holdings and instability arising 
	from fluctuations in asset prices.  More specifically, we show that 
	there are networks where small fluctuations in cross-holdings or asset prices can have striking consequences,
	and these consequences have numerous implications.

	Our first result (Corollary \ref{cor:acyclicsensitivity}, Theorems \ref{thm:cyclicsensitivity}, \ref{thm:cycliclowerbound}) shows 
	that the effect of small changes in cross-holdings is strongly tied to the integration of the network.
	In highly integrated networks small changes in cross-holdings can have potentially unbounded effects 
	on market valuations, while in networks with low integration, changes in cross-holdings have more tightly 
	bounded effects on the market values of the institutions in the network.  These results can be interpreted in 
	many ways.  From a regulatory perspective, in a highly integrated network, any regulator must know the entire cross-holdings network 
	to a very high degree of accuracy in order effectively understand the market values of the institutions.
	From an institution's perspective, small changes in investment by individual institutions can have their effects greatly 
	magnified throughout the network.  From a predictivity perspective, if one wishes to forecast market values into the future, 
	any small forecasting uncertainty in the cross-holdings can have enormous effects on the (predicted) market values of the institutions.
	From a privacy perspective, institutions cannot maintain any privacy in their investment portfolios without compromising the ability of 
	outsiders (e.g. other institutions, outside investors or regulators) to calculate the market value of the institutions.
	These problems arise only in highly integrated networks (Theorem \ref{thm:cyclicsensitivity}), and they can all be mitigated by imposing a cap 
	on integration.  If institutions are required to maintain some fixed percentage of their ownership outside of the network, 
	the sensitivity to changes in cross-holdings can be drastically reduced.
	
	Our second result shows that small changes in the prices of the underlying assets can have unpredictable 
	effects on the number of failures in the system.  Specifically, we show that there are networks where 
	it is \emph{computationally intractable}, even with perfect information about the cross-holdings, to 
	estimate the number of failures that can occur after some small drop in asset prices.
	This result too can be interpreted from different perspectives.
	This result implies that a regulator (with perfect information about the network cross-holdings) 
	who believes there may be some bounded fluctuation in asset prices cannot be expected to distinguish a network 
	where these fluctuations cause a small number of failures, from a network where fluctuations of the same magnitude 
	can cause a massive number of failures.  Institutions face the same computational challenge.  An institution may 
	wonder whether its investment portfolio will protect it from a bounded shock in asset prices, and our results show 
	that even with complete information about the investments of all other agents in the system, it may be infeasible 
	to determine whether a specific portfolio is safe.

	Previous works have imposed specific probabilistic models on fluctuations in asset prices, and for a given probabilistic model 
	the number of failures can usually be estimated.  But what if the model is incorrect (e.g. the asset prices do not fluctuate independently)?
	Our results show that there are situations where changes in the distribution of fluctuations (but not their overall magnitude) can 
	have huge and unpredictable effects on the number of failures.

	Moving forward, it is an important research question to understand what constraints on the network will allow 
	institutions and regulators to perform these stability analyses.

%% file: appendix.tex
\appendix
	
	\section{A flow model of the system}
	\label{app:flow}

	It is often instructive to create a stochastic matrix representing the system, which models 
	money flowing through the network.

	There are $n$ financial institutions, and we introduce $n$ addition entities, the ``shareholders.''
	In this model shareholder $i$ owns a $\hat{C}_{ii}$ fraction of institution $i$ and has no other financial ties to the system.
	This defines a $2n \times 2n$ matrix, $\mm{A}$, defined as 
		
	\[
		\mm{A} = \left[ \begin{array}{c|c} \mm{C} & \mm{0} \\ \hline \hat{\mm{C}} & \mm{I} \end{array} \right]
	\]

	The augmented matrix $\mm{A}$ is now column stochastic, \ie its columns sum to $1$.
	The $ij$th entry of the matrix $\mm{A}$ represents the fraction of ``agent'' $j$ owned by agent $i$.
	There are $2n$ agents because there are $n$ financial institutions (agents $1,\ldots,n$) and $n$ collections of shareholders (agents $n+1,\ldots,2n$).  
	The lower right corner of $\mm{A}$ is the identity matrix because the shareholders are completely self-owned.  
	The columns of $\mm{A}$ sum to one because for each agent, $j$, its entire value is owned by the other financial institutions or 
	the external shareholders.
	
	Initially, the banks are assumed to have some intrinsic value $\mm{D} \vv{p}$, which indicates the value of underlying assets owned by each institution.
	Without loss of generality, we will assume that all assets are completely owned by the institutions in the network (otherwise, we can simply rescale the value of each asset).
	This assumption is equivalent to saying that $\mm{D}$ is column stochastic.

	The $2n \times 2n$ matrix $\mm{A}$ allows us to view the market valuations of each institution as the steady state of a dynamical process.
	Money flows into the system from the underlying assets, and at each time step, the value residing in each financial institution is distributed 
	to its stakeholders according to their stake.  This process terminates when all the money in the system (coming from the underlying assets) has 
	been distributed to the external shareholders.  Algebraically, this process can be viewed as follows:
	given a vector $\vv{W} \in \R^{2n}$, where $W_i$ denotes the value of the underlying assets owned by $i$, for $i=1,\ldots,n$ and $W_i = 0$ for $n < i \le 2n$, 
	then 
	\begin{equation}
		\lim_{t \rightarrow \infty} \mm{A}^t \vv{W} = \lim_{t \rightarrow \infty} \mm{A}^t \left[\begin{array}{c} \vv{w} \\ \hline \vv{0} \end{array} \right]  = \left[\begin{array}{c} \vv{0} \\ \hline \vv{v} \end{array} \right] = \left[ \begin{array}{c} 0 \\ \vdots \\ 0 \\ \hline v_{1} \\ \vdots \\ v_n \end{array} \right]
		\label{eqn:alimit}
	\end{equation}
	where $v_i$ will denote the market value of institution $i$, \ie $v_i$ is the value of institution $i$ owned by its external shareholders.
	Now, basic linear algebra tells us that 
	\[
		\mm{A}^t = \left[ \begin{array}{c|c} \mm{C}^t & \mm{0} \\ \hline \hat{\mm{C}}\left(\mm{C}^{t-1} + \mm{C}^{t-2} + \cdots + \mm{I} \right) & \mm{I} \end{array} \right]
	\]

	Throughout this work, we will use the standard operator norm for a matrix (in the $L_1$ sense),
	\[
		\| \mm{A} \| = \| \mm{A} \|_1 = \sup_{\vv{x}} \frac{ \| \mm{A} \vv{x} \|_1 }{\| \vv{x} \|_1 }.
	\]

	Because the columns of $\mm{A}$ sum to one, we have $\| \mm{A} \vv{x} \|_1 =  \|\vv{x}\|_1$
	for any $\vv{x}$.  In particular, this means that 
	\begin{equation}
		\label{eqn:noloss}
		\left\| \mm{A} \left[ \begin{array}{c} \mm{D} \vv{p} \\ \vec{0} \end{array} \right] \right\|_1= \left\| \vv{p} \right\|_1
	\end{equation}
	This is the algebraic statement that the total market value of all the institutions is exactly the total value of the underlying assets in the system.
	
	To analyze equation \ref{eqn:alimit} we recall the standard fact about matrix series
	\begin{lemma}
		\label{lem:matrixsum}

		If $\mm{C}$ is a matrix with $\|\mm{C}\| < 1$, then $\mm{I} - \mm{C}$ is invertible and 
		\[
			(\mm{I} - \mm{C})^{-1} = \sum_{k=0}^\infty \mm{C}^k
		\]
	\end{lemma}

	\begin{proof}
		First, note if $\vv{v} \ne \vv{0}$, then $\|\mm{C}\vv{v}\| < \|\vv{v}\|$, so 
		\[
			\| (\mm{I} - \mm{C})\vv{v} \| \ge \| \mm{I} \vv{v} \| - \|\mm{C}\vv{v}\| > 0
		\]
		so $\mm{I} - \mm{C}$ has a trivial kernel, and hence is invertible.

		Then, the result follows just as in the scalar case:
		Letting $\mm{S} = \sum_{k=0}^N \mm{C}^k$, we have $\mm{S}( \mm{I} - \mm{C} ) = \mm{I} - \mm{C}^{N+1}$.
		Thus
		\[
			\sum_{k=0}^N \mm{C}^k = (\mm{I} - \mm{C})^{-1} (\mm{I} - \mm{C}^{N+1})
		\]
		Letting $N \rightarrow \infty$, and noting that $\mm{C}^{N+1} \rightarrow 0$ gives the result.
	\end{proof}

	This leads to the closed formula used in \cite{EGJ14}:
	\begin{equation}
		\label{eqn:marketvalue}
		\vv{v} = \hat{\mm{C}} ( \mm{I} - \mm{C} )^{-1} \mm{D}\vv{p}
	\end{equation}
	Where the $n \times 1$ vector $\vv{p}$ represents the values of the underlying assets and the vector $\vv{v}$ represents the market value of the financial institutions (as measured by their external shareholders).
	Equations \ref{eqn:alimit} and \ref{eqn:noloss} then tell us that $\|\mm{D}\vv{p}\|_1 = \|\vv{p}\|_1 = \|\vv{v}\|_1$, which is just the simple statement that the total value of the agents is the same as the total value of 
	the underlying assets, in other words value is never created or destroyed by the network.